\documentclass[11pt]{amsart}
\usepackage[margin=1in]{geometry}
\usepackage[utf8]{inputenc}
\usepackage{amsmath,amssymb,amsthm,mathtools}
\usepackage{enumitem}
\usepackage{tikz}
\usepackage{flafter,placeins}
\usetikzlibrary{arrows.meta,decorations.pathreplacing}
\definecolor{figblue}{RGB}{34,85,130}
\definecolor{figorange}{RGB}{180,90,30}
\definecolor{figgray}{RGB}{100,100,100}
\usepackage{hyperref}
\hypersetup{colorlinks=true,linkcolor=blue,citecolor=blue,urlcolor=blue}
\theoremstyle{definition}
\newtheorem{definition}{Definition}[section]
\theoremstyle{plain}
\newtheorem{theorem}[definition]{Theorem}
\newtheorem{proposition}[definition]{Proposition}
\newtheorem{lemma}[definition]{Lemma}
\newtheorem{corollary}[definition]{Corollary}
\theoremstyle{definition}
\newtheorem{example}[definition]{Example}
\theoremstyle{remark}
\newtheorem{remark}[definition]{Remark}
\DeclareMathOperator{\arctanh}{arctanh}
\DeclareMathOperator{\br}{br}
\DeclareMathOperator{\Lip}{Lip}
\DeclareMathOperator{\dist}{dist}

\title[Branching Gaps and Ising Coexistence]{Branching Gaps, Coexistence, and Boundary Influence for the Ising Model on Spherically Symmetric Trees}

\author{Farrukh Mukhamedov}
\address{Department of Mathematical Sciences, College of Science, The United Arab Emirates University, P.O. Box 15551, Al Ain, Abu Dhabi, UAE}
\email{far75m@yandex.ru; farrukh.m@uaeu.ac.ae}
\author{Otabek Khakimov}
\address{V.I. Romanovskiy Institute of Mathematics, Uzbekistan Academy of Sciences, 9, University str., 100174, Tashkent, Uzbekistan}
\email{khakimovo86@gmail.com}

\begin{document}
\begin{abstract}
In this paper, we study the ferromagnetic nearest-neighbor Ising model on rooted
spherically symmetric trees with one or \(q\) children per vertex,
according to a prescribed sequence of branching levels. Unbounded
runs of one-child levels force uniqueness at every nonzero homogeneous
field and every finite inverse temperature. We quantify this effect
through the locations of sufficiently long nonbranching stretches.
Conversely, uniqueness need not entail any uniform rate: for each
prescribed vanishing sequence, we construct a density-one tree whose
finite-volume boundary influence decays more slowly along a subsequence.
The tree retains branching number \(q\), although its nonzero-field
coexistence region disappears. For eventually bounded gaps, we prove
an optimal uniform comparison with the periodic tree having the
largest allowed gap and obtain an explicit interval of coexistence
fields. Thus bounded gaps are equivalent, within this family, to
coexistence at some finite temperature and nonzero field. For eventually
periodic gaps, Gibbs uniqueness is equivalent to fixed-point uniqueness
of the full-period return map. A negative Schwarzian derivative yields
a complete classification: in the supercritical regime the map has
three fixed points inside a closed coexistence interval, two at its
endpoints, and one outside. The nonzero endpoints are automatically
nondegenerate folds, with square-root splitting of the merging branches.
\end{abstract}
\maketitle

\section{Introduction}
In the ferromagnetic Ising model, branching growth is closely linked to the zero-field transition. A homogeneous nonzero field raises the question of whether border influence may persist without branching, even while the exponential growth rate remains constant. We answer this question for leafless, spherically symmetric trees with offspring numbers in \(\{1,q\}\), where \(q\ge2\). In this study, uniqueness refers to the uniqueness of all Dobrushin-Lanford-Ruelle (DLR) Gibbs measures.

Let \(\alpha=\{n_1<n_2<\cdots\}\subset\mathbb Z_{\ge0}\) be the
branching levels, let \(\ell_j=n_{j+1}-n_j-1\) count the intervening
one-child levels, and put \(\theta=\tanh(\beta J)\) and \(H=\beta B\).
The results develop in three directions.
\begin{enumerate}[label=\textup{(\roman*)}]
\item \emph{Geometric uniqueness and a matching coexistence comparison.}
Theorem~\ref{thm:unbounded-gaps} proves uniqueness at every \(H\ne0\)
when \(\sup_j\ell_j=\infty\). For \(H>0\), a long gap first makes
the limiting minus field positive; a logarithmic contraction identifies the
extremal Gibbs measures. For eventually \(L\)-bounded gaps,
Theorem~\ref{thm:bounded-gap-comparison} proves coexistence wherever
the constant-gap-\(L\) tree has coexistence. This is the optimal
guarantee uniform over that class.
Corollary~\ref{cor:bounded-gap-explicit} gives an explicit positive
field interval when \(q\theta^{L+1}>1\), and
Corollary~\ref{cor:gap-dichotomy} completes the geometric dichotomy
for existence of nonzero-field coexistence at some temperature.

\item \emph{Quantitative boundary influence and its limitations.}
Theorem~\ref{thm:quantitative-gap} bounds finite-volume boundary
influence using the locations of sufficiently long gaps.
Theorem~\ref{thm:equal-density} constructs two trees of branching
density one and branching number \(q\) with different nonzero-field
coexistence sets. Theorem~\ref{thm:arbitrarily-slow} strengthens this
comparison: for each prescribed vanishing rate and fixed parameters
in the regular tree's nonzero-field coexistence region, a density-one
tree can have uniqueness and slower boundary-influence decay along
a subsequence. Thus growth invariants determine neither field
coexistence nor a uniform rate of boundary forgetting.

\item \emph{A complete periodic phase diagram.}
Theorem~\ref{thm:periodic-return} identifies the extremal branch fields
with the least and greatest fixed points of the full-period return
map. Theorem~\ref{thm:periodic-phase} then classifies all of its fixed
points. If a period has \(r\) branching levels and \(P\) levels in
total, the zero-field threshold is \(q^r\theta^P=1\). Above this
threshold, there are exactly three fixed points inside the coexistence
interval, two at its endpoints, and one outside. The nonzero-field
endpoints are nondegenerate folds. These counts concern periodic
boundary laws; the uniqueness criterion applies to all Gibbs measures.
\end{enumerate}

\subsection*{Relation to earlier results}
The boundary-law description and connection between Markov chains and Gibbs fields on trees are traditional \cite{Preston1974, Spitzer1975, Zachary1983, Georgii2011}.
Lyons \cite{Lyons1989} found the zero-field critical threshold using the branching number. Pemantle and Peres \cite[Corollary~2.4]{PemantlePeres2010} define the critical criteria as the convergence of \(\sum_{N\ge1}\theta^{-2N}q^{-2b_N}\), where \(b_N\) represents branching levels below \(N\). Appendix~\ref{sec:zero-field-appendix} covers the straightforward scalar proof and density ramifications.

Jonasson and Steif ~\cite{JonassonSteif1999} demonstrate coexistence and couplings on nonamenable graphs with restricted degrees, as well as uniqueness at nonzero fields on amenable quasi-transitive graphs. These trees contain unbounded gaps, allowing for extended routes of degree-two vertices with bounded edge boundaries. These trees are not quasi-transitive since the lengths between branching vertices are unlimited. The quasi-transitive uniqueness theorem alone does not result in Theorem~\ref{thm:unbounded-gaps}. In contrast, constrained gaps result in a limited subdivision of a regular rooted tree. The qualitative existence assertion in Corollary~\ref{cor:gap-dichotomy} aligns with the nonamenability hypothesis. The quantitative assertions include the explicit field bound and optimum comparison to the largest-gap periodic tree.

Geometry-dependent nonzero-field phase diagrams are not new. Schonmann and Tanaka \cite{SchonmannTanaka1998} found that attaching a predetermined number of leaves to each vertex of a regular tree resulted in a much smaller coexistence area. With finite leaf ornamentation, the number of branches remains fixed. Our equal-density comparison applies to leafless spherically symmetric \(\{1,q\}\)-offspring trees. Both branching densities remain one, and a zero-density set of nonbranching levels removes the nonzero-field coexistence area. The comparison maintains the growth invariants but alters amenability. The slow-convergence construction distinguishes infinite-volume uniqueness from finite-volume rates.

H\"aggstr\"om, Schonmann and Steif~\cite{HaggstromSchonmannSteif2000}
study random dilution, including persistence of nonzero-field
transitions under sufficiently small Bernoulli dilution of nonamenable
graphs of bounded degree. Our suppressed branching levels are
deterministic and synchronized across each generation. Bissacot, Endo
and van Enter~\cite{BissacotEndoVanEnter2017} study level-dependent
fields approaching a critical field on a fixed regular Cayley tree.
Here the applied field and coupling are homogeneous, and the level
degrees vary. The quantitative observable is the difference of the
finite-volume root magnetizations under plus and minus boundary
conditions.

The return-map framework employs standard boundary-law approaches. The global fixed-point classification uses the negative Schwartz derivative of the basic recursion and its behavior during composition. The Schwarzian identities are widely used in one-dimensional dynamics \cite{deMeloVanStrien1993}, and we present the necessary quick computations here. This results in a precise count and automatic nondegeneracy of nonzero-field folds, beyond the fixed-point reduction.

Periodic compositions and distinctions between uniqueness and
correlation decay also occur for general ferromagnetic two-spin
systems. Guo and Lu~\cite[Section~5]{GuoLu2016} analyze an
alternating-offspring tree by composing its level recursion maps.
Here the applied Ising field is homogeneous, and the periodic
classification is complemented by a density-one construction with
uniqueness and an arbitrarily prescribed slow rate along a subsequence.

Section~\ref{sec:setup} fixes the notation and classical benchmarks.
Section~\ref{sec:unbounded} proves the unbounded-gap theorem and the
equal-density comparison. Section~\ref{sec:boundary-influence} gives
the quantitative bounds and the slow-convergence construction.
Section~\ref{sec:periodic} treats periodic gaps, and
Section~\ref{sec:bounded-gaps} proves the bounded-gap comparison.
Section~\ref{sec:sufficient} records supplementary criteria.

\section{Setup and Preliminaries}
\label{sec:setup}

\subsection{Tree geometry and the specification}
Fix an integer \(q\ge2\) and an infinite increasing branching sequence
\(\alpha=\{n_j\}_{j\ge1}\subset\mathbb Z_{\ge0}\).
In the rooted tree \(T_{\alpha,q}=(V,E)\), every vertex at distance
\(m\) from the root \(o\) has
\[
a_m=\begin{cases}q,&m\in\alpha,\\1,&m\notin\alpha\end{cases}
\]
children. Write \(W_m=\{x:\dist(o,x)=m\}\),
\(\Lambda_m=\bigcup_{k=0}^mW_k\), and
\[
b_N=\#(\alpha\cap\{0,\ldots,N-1\}),\qquad
\underline d=\liminf_{N\to\infty}\frac{b_N}{N},\qquad
\overline d=\limsup_{N\to\infty}\frac{b_N}{N}.
\]
Successive multiplication by the offspring numbers gives
\(|W_N|=\prod_{m<N}a_m=q^{b_N}\). The standard spherically
symmetric branching-number formula \cite{LyonsPeres2016} yields
\begin{equation}\label{eq:branching-number}
\br(T_{\alpha,q})=\liminf_{N\to\infty}|W_N|^{1/N}
=q^{\underline d}.
\end{equation}
Indeed, a unit flow assigning mass \(|W_N|^{-1}\) to each vertex of
\(W_N\) meets the capacity bound \(C\lambda^{-N}\) for every
\(\lambda<\liminf|W_N|^{1/N}\), while level cutsets give the reverse
inequality. We also write \(\underline d(\alpha),\overline d(\alpha)\)
when comparing trees.

Let \(J>0\), \(0<\beta<\infty\), and \(B\in\mathbb R\), and put
\[
K=\beta J,\qquad \theta=\tanh K\in(0,1),\qquad H=\beta B.
\]
The finite-volume Gibbs law on \(\Lambda_{N-1}\) with spins
\(\eta\) fixed on \(W_N\) is proportional to
\[
\exp\!\left\{
K\!\sum_{\substack{\{x,y\}\in E\\x,y\in\Lambda_{N-1}}}\sigma_x\sigma_y
+K\!\sum_{\substack{\{x,y\}\in E\\x\in\Lambda_{N-1},\ y\in W_N}}
\sigma_x\eta_y+H\sum_{x\in\Lambda_{N-1}}\sigma_x\right\}.
\]
Write \(\mu_{N,B}^\pm\) for uniform plus/minus boundary conditions
and \(\mu_B^\pm\) for their local weak limits. Ferromagnetic
attractiveness gives their existence and
\(\mu_B^-\preceq\mu\preceq\mu_B^+\) for every DLR measure \(\mu\)
\cite{Georgii2011}. Consequently, Gibbs uniqueness is equivalent to
\(\mu_B^- =\mu_B^+\).

\subsection{Cavity fields and gap maps}
Define
\begin{equation}\label{F_theta}
F_\theta(u)=\arctanh(\theta\tanh u),\qquad
F_\theta(\pm\infty)=\pm K.
\end{equation}
It is odd, strictly increasing and bounded, with
\begin{equation}\label{eq:derivative}
F_\theta'(u)=\frac{\theta(1-\tanh^2u)}{1-\theta^2\tanh^2u}.
\end{equation}
This derivative is even, is at most \(\theta<1\), and strictly
decreases on \((0,\infty)\).

\begin{proposition}
\label{prop:field-recursion}
The spherically symmetric cavity recursion is
\begin{equation}\label{eq:cavity}
u_m=H+a_mF_\theta(u_{m+1}),\qquad m\ge0.
\end{equation}
For uniform boundary conditions at depth \(N\), its terminal values
are \(u_N^{\pm,N}=\pm\infty\). For each fixed \(m\), the plus fields
decrease and the minus fields increase as \(N\to\infty\), with finite
limits \(u_m^\pm\) satisfying \eqref{eq:cavity} and
\[
H-a_mK\le u_m^-\le u_m^+\le H+a_mK.
\]
The model is unique if and only if \(u_m^-=u_m^+\) at one, and hence
every, level \(m\).
\end{proposition}

\begin{proof}
A child subtree of cavity field \(u\) contributes
\(\sum_{t=\pm1}e^{Kst+ut}=2\cosh(Ks+u)\) to its parent spin \(s\).
The half log-ratio for \(s=+1\) and \(s=-1\) is \(F_\theta(u)\),
giving \eqref{eq:cavity}. Comparing the terminal field \(+\infty\)
with \(H+a_NK\), and \(-\infty\) with \(H-a_NK\), proves the depth
monotonicity after composition with the increasing recursion maps.
The bounds follow from \(|F_\theta|\le K\), so the limits exist
and satisfy the recursion.

The limiting root distribution and child transition probabilities are
\begin{equation}\label{eq:markov-reconstruction}
\mu_B^\pm(\sigma_o=s)=\frac{e^{u_0^\pm s}}{2\cosh u_0^\pm},\qquad
\pi_m^\pm(t\mid s)=
\frac{e^{Kst+u_{m+1}^\pm t}}{2\cosh(Ks+u_{m+1}^\pm)}.
\end{equation}
They are the limits of the finite-volume kernels and reconstruct
\(\mu_B^\pm\). Each map \(u\mapsto H+a_mF_\theta(u)\) is injective,
so equality at one level implies equality at the next level and,
by backward recursion, at all preceding levels. Thus equality at one
level identifies the whole field profiles and the two extremal Gibbs
measures. Conversely, equality of those measures implies equality of
their root fields, and injectivity then gives equality at every level.
\end{proof}

More generally, every finite real profile satisfying \eqref{eq:cavity}
defines a splitting (tree-indexed Markov) Gibbs measure through
\eqref{eq:markov-reconstruction}, without the superscripts \(\pm\).
The recursion cancels the child normalizing factors to give the local
DLR conditional probabilities; see also
\cite{Spitzer1975,Zachary1983}. No assertion that all Gibbs measures
are splitting is required for the extremal uniqueness criterion.

Every field obtained after at least one recursion step lies in
\(I_H=[H-qK,H+qK]\). Set
\begin{equation}\label{eq:blockmaps}
g_H(u)=H+F_\theta(u),\qquad Q_H(u)=H+qF_\theta(u),\qquad
\Phi_{\ell,H}=Q_H\circ g_H^{\,\ell}.
\end{equation}
Here \(g_H^0\) is the identity. With
\(\ell_j=n_{j+1}-n_j-1\) and \(v_j=u_{n_j}\), the recursion becomes
\begin{equation}\label{eq:branch-recursion}
v_j=\Phi_{\ell_j,H}(v_{j+1}).
\end{equation}
Thus a gap counts only the intervening one-child levels; the distance
between the branching levels is \(\ell_j+1\).
Figure~\ref{fig:gap-geometry} illustrates the level convention and the
order of the updates in the block map.

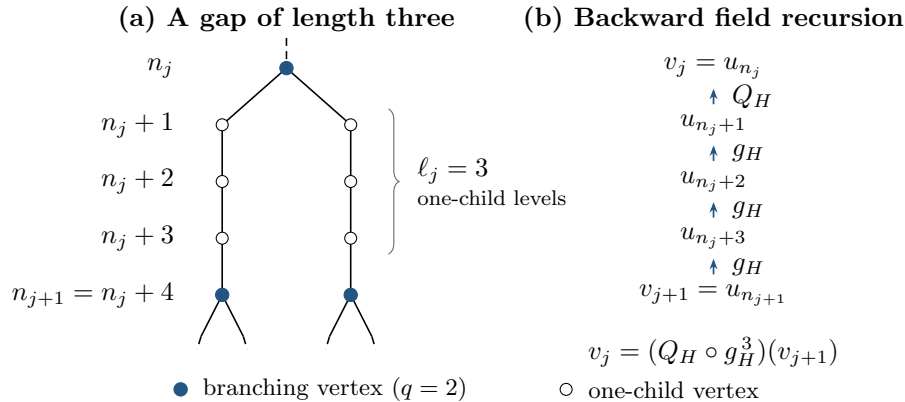
\begin{figure}[htbp]
\centering
\begin{tikzpicture}[
  x=1cm,y=1cm,font=\small,
  branch/.style={circle,draw=figblue,fill=figblue,inner sep=0pt,minimum size=5pt},
  unary/.style={circle,draw=black,fill=white,inner sep=0pt,minimum size=4.5pt},
  tree edge/.style={line width=.6pt},
  message/.style={-{Stealth[length=4pt]},line width=.65pt,draw=figblue}
]
\node[font=\small\bfseries] at (0,.65) {(a) A gap of length three};
\node[font=\small\bfseries] at (5.65,.65) {(b) Backward field recursion};

\draw[tree edge,densely dashed] (0,.4)--(0,.08);
\foreach \s in {-1,1} {
  \draw[tree edge] (0,0)--({.85*\s},-.75)
    --({.85*\s},-1.5)--({.85*\s},-2.25)
    --({.85*\s},-3);
  \draw[tree edge] ({.85*\s},-3)--({.85*\s-.28},-3.55);
  \draw[tree edge] ({.85*\s},-3)--({.85*\s+.28},-3.55);
  \draw[tree edge,densely dashed] ({.85*\s-.28},-3.55)
    --({.85*\s-.33},-3.72);
  \draw[tree edge,densely dashed] ({.85*\s+.28},-3.55)
    --({.85*\s+.33},-3.72);
  \foreach \y in {-.75,-1.5,-2.25}
    \node[unary] at ({.85*\s},\y) {};
  \node[branch] at ({.85*\s},-3) {};
}
\node[branch] at (0,0) {};
\node[anchor=east] at (-1.35,0) {$n_j$};
\node[anchor=east] at (-1.35,-.75) {$n_j+1$};
\node[anchor=east] at (-1.35,-1.5) {$n_j+2$};
\node[anchor=east] at (-1.35,-2.25) {$n_j+3$};
\node[anchor=east] at (-1.35,-3) {$n_{j+1}=n_j+4$};
\draw[decorate,decoration={brace,amplitude=4pt},draw=figgray]
  (1.35,-.54)--(1.35,-2.46);
\node[anchor=west,align=left] at (1.6,-1.5)
  {$\ell_j=3$\\[-1pt]\scriptsize one-child levels};

\node (f0) at (5.65,0) {$v_j=u_{n_j}$};
\node (f1) at (5.65,-.75) {$u_{n_j+1}$};
\node (f2) at (5.65,-1.5) {$u_{n_j+2}$};
\node (f3) at (5.65,-2.25) {$u_{n_j+3}$};
\node (f4) at (5.65,-3) {$v_{j+1}=u_{n_{j+1}}$};
\draw[message] (f1.north)--node[right=3pt] {$Q_H$} (f0.south);
\draw[message] (f2.north)--node[right=3pt] {$g_H$} (f1.south);
\draw[message] (f3.north)--node[right=3pt] {$g_H$} (f2.south);
\draw[message] (f4.north)--node[right=3pt] {$g_H$} (f3.south);
\node at (5.65,-3.75) {$v_j=(Q_H\circ g_H^{\,3})(v_{j+1})$};
\node[branch] at (-1.4,-4.25) {};
\node[anchor=west,font=\footnotesize] at (-1.24,-4.25)
  {branching vertex ($q=2$)};
\node[unary] at (3.7,-4.25) {};
\node[anchor=west,font=\footnotesize] at (3.86,-4.25)
  {one-child vertex};
\end{tikzpicture}
\caption{Branching gaps and the order of composition. Panel (a) shows
the descendants of one vertex at level $n_j$, with $q=2$ and
$\ell_j=3$. The gap counts the three intervening one-child levels;
the distance between branching levels is four. Panel (b) follows a
spherically symmetric cavity field from level $n_{j+1}$ toward the
root. Three $g_H$ updates precede $Q_H$, so
$\Phi_{3,H}=Q_H\circ g_H^3$.}
\label{fig:gap-geometry}
\end{figure}
\FloatBarrier

\subsection{Known zero-field and field-monotonicity results}
\begin{theorem}
\label{thm:exact}
At \(H=0\), there is more than one Gibbs measure if and only if
\begin{equation}\label{eq:zero-field-series}
\sum_{N=1}^{\infty}\theta^{-2N}q^{-2b_N}<\infty.
\end{equation}
\end{theorem}
This is \cite[Corollary~2.4]{PemantlePeres2010} specialized to
\(|W_N|=q^{b_N}\). Appendix~\ref{sec:zero-field-appendix} gives a
short self-contained proof. Away from equality, coexistence occurs
when \(\theta q^{\underline d}>1\) and uniqueness when
\(\theta q^{\underline d}<1\). At equality the full counting sequence
enters \eqref{eq:zero-field-series}.

For fixed temperature, uniqueness is monotone in the magnitude of a
homogeneous field on every locally finite graph
\cite[Proposition~2]{SchonmannTanaka1998}: uniqueness at \(H_1\)
implies uniqueness whenever \(|H|>|H_1|\). In particular, zero-field
uniqueness implies uniqueness at every homogeneous field. A more
general boundary-influence comparison allowing arbitrary fields is
proved in \cite{DingSongSun2023}. We use the homogeneous statement
as an established result.

Spin flip and this monotonicity show that
\(\mathcal C_\theta=\{H\in\mathbb R:\mu_{H/\beta}^-\ne
\mu_{H/\beta}^+\}\) is a symmetric interval, possibly empty or
\(\{0\}\). Define
\begin{equation}\label{eq:Hc}
H_c(\theta;\alpha,q)=\sup\bigl((\mathcal C_\theta\cap[0,\infty))
\cup\{0\}\bigr).
\end{equation}
It is finite by Proposition~\ref{thm:high-field-uniqueness} below.
This definition does not settle coexistence at the endpoints.

\section{Unbounded gaps and the equal-density comparison}
\label{sec:unbounded}

The first result uses only the existence of arbitrarily long gaps.
Their frequency and the asymptotic density of branching play no role.

\begin{theorem}
\label{thm:unbounded-gaps}
Suppose that the non-branching gap lengths are unbounded:
\[
\sup_{j\ge1}\ell_j=\infty.
\]
Then, for every finite inverse temperature \(\beta>0\) and every
homogeneous field \(H=\beta B\ne0\), the Ising model on
\(T_{\alpha,q}\) has a unique Gibbs measure. In particular,
\(H_c(\theta;\alpha,q)=0\), irrespective of whether the zero-field
model is unique.
\end{theorem}

\begin{proof}
By spin-flip symmetry it suffices to consider \(H>0\). Recall that
\(K=\arctanh\theta\) and that the finite limiting plus and minus fields
satisfy
\[
u_m^\pm=H+a_mF_\theta(u_{m+1}^\pm),
\qquad
H-qK\le u_m^-\le u_m^+\le H+qK.
\]
Since \(0<F_\theta'(u)\le\theta<1\), the map
\(g_H(u)=H+F_\theta(u)\) is a contraction of \(\mathbb R\).
Its unique fixed point \(r_H\) is positive: the function
\(u-F_\theta(u)\) is strictly increasing and vanishes at zero, whereas
\(r_H-F_\theta(r_H)=H>0\). After the first iterate from \(-\infty\),
the orbit is finite, and the contraction theorem gives
\[
g_H^\ell(-\infty)\longrightarrow r_H>0.
\]
Thus there is an integer \(L\ge1\) such that
\(g_H^\ell(-\infty)>0\) for every \(\ell\ge L\).

Whenever \(\ell_j\ge L\), monotonicity and the block recursion imply
\[
u_{n_j}^-
=Q_H\!\left(g_H^{\ell_j}(u_{n_{j+1}}^-)\right)
\ge Q_H\!\left(g_H^{\ell_j}(-\infty)\right)>H.
\]
The backward recursion preserves the lower bound \(H\), since
\(H+aF_\theta(u)\ge H\) for \(u\ge H>0\) and
\(a\in\{1,q\}\). Unboundedness of the finite gap lengths supplies
indices \(j\) with \(\ell_j\ge L\) arbitrarily far from the root.
Consequently, for every level \(m\),
\begin{equation}
\label{eq:positive-fields-unbounded-gaps}
H\le u_m^-\le u_m^+\le H+qK.
\end{equation}

On this positive interval, use logarithmic coordinates. For
\(a\in\{1,q\}\), define
\[
f_a(x):=\log\!\left(H+aF_\theta(e^x)\right).
\]
The decrease of \(F_\theta'\) on \((0,\infty)\), together with
\(F_\theta(0)=0\), gives
\[
uF_\theta'(u)\le\int_0^u F_\theta'(t)\,dt=F_\theta(u)
\qquad (u>0).
\]
Hence
\[
0<f_a'(x)
=\frac{ae^xF_\theta'(e^x)}{H+aF_\theta(e^x)}
\le\frac{aF_\theta(e^x)}{H+aF_\theta(e^x)}
\le\frac{qK}{H+qK}
=:c_H<1.
\]
Applying this bound to the logarithms of the limiting fields and
iterating from any level \(N>m\) back to \(m\), we obtain
\[
\begin{aligned}
0\le\log u_m^+-\log u_m^-
&\le c_H^{N-m}
       \bigl(\log u_N^+-\log u_N^-\bigr)\\
&\le c_H^{N-m}\log\frac{H+qK}{H},
\end{aligned}
\]
where the last inequality follows from
\eqref{eq:positive-fields-unbounded-gaps}. Letting \(N\to\infty\)
yields \(u_m^+=u_m^-\) for every \(m\). The root laws and transition
kernels therefore coincide, so \(\mu_B^+=\mu_B^-\), which is equivalent
to uniqueness among all DLR Gibbs measures. The assertion about
\(H_c\) follows from its definition.
\end{proof}

\begin{corollary}
\label{cor:bounded-gaps-necessary}
If the model has more than one Gibbs measure at some homogeneous field
\(H\ne0\), then \(\sup_{j\ge1}\ell_j<\infty\).
\end{corollary}

\begin{proof}
This is the contrapositive of Theorem~\ref{thm:unbounded-gaps}.
\end{proof}

\begin{theorem}
\label{thm:equal-density}
For every \(q\ge2\), define
\[
t_k=2^{k^2},\qquad D_k=\{t_k,\ldots,t_k+k-1\},\qquad
\alpha=\mathbb Z_{\ge0}\setminus\bigcup_{k\ge1}D_k.
\]
The trees \(T_{\alpha,q}\) and the regular rooted \(q\)-ary tree
both have lower and upper branching densities equal to one and
branching number \(q\). For every \(\theta\in(1/q,1)\), their
coexistence sets are respectively
\[
\mathcal C_\theta(T_{\alpha,q})=\{0\},\qquad
\mathcal C_\theta(T_{\mathrm{reg},q})
=[-H_c^{\mathrm{reg}}(\theta,q),H_c^{\mathrm{reg}}(\theta,q)],
\]
where \(H_c^{\mathrm{reg}}(\theta,q)>0\) is given in
Proposition~\ref{prop:regular-field}.
\end{theorem}

\begin{proof}
The blocks \(D_k\) are pairwise disjoint and separated by branching
levels: \(t_k+k<t_{k+1}\) for every \(k\ge1\).  In particular,
\(t_k-1\) and \(t_k+k\) are consecutive elements of \(\alpha\), with
exactly \(k\) non-branching levels between them.  Hence the gap lengths
\(\ell_j\) are unbounded.

Let \(r_N:=N-b_N\) be the number of deleted levels below \(N\).  For
\(N\ge2\), put \(m_N:=\lfloor\sqrt{\log_2 N}\rfloor\).  A block
that meets \(\{0,\ldots,N-1\}\) must have \(t_k<N\), and therefore
\(k\le m_N\).  Thus
\[
0\le r_N\le\sum_{k=1}^{m_N}k
=\frac{m_N(m_N+1)}2
=O(\log(N+2)).
\]
It follows that
\[
b_N=N-O(\log(N+2)),\qquad
\underline d(\alpha)=\overline d(\alpha)=1.
\]
Moreover,
\[
\theta^{-2N}q^{-2b_N}
=(q\theta)^{-2N}q^{2r_N},
\]
where \(q^{2r_N}\) grows at most polynomially in \(N\).  Since
\(q\theta>1\), the zero-field series converges, and
Theorem~\ref{thm:exact} gives coexistence at \(H=0\).  On the other
hand, Theorem~\ref{thm:unbounded-gaps} gives uniqueness at every
\(H\ne0\).  Consequently,
\[
H_c(\theta;\alpha,q)=0,
\]
even though the coexistence set is nonempty.

For comparison, let \(\alpha_{\mathrm{reg}}=\mathbb Z_{\ge0}\).
The rooted \(q\)-ary tree \(T_{\alpha_{\mathrm{reg}},q}\) has the
same lower and upper branching densities, both equal to one, and the
same branching number \(q\).  By
Theorem~\ref{thm:periodic-return} and Proposition~\ref{prop:regular-field},
the regular-tree model has coexistence precisely for
\[
|H|\le H_c^{\mathrm{reg}}(\theta,q),
\qquad H_c^{\mathrm{reg}}(\theta,q)>0.
\]
The endpoints are included because the regular-tree equation has two
distinct fixed points there.  Hence
\[
H_c(\theta;\alpha_{\mathrm{reg}},q)
=H_c^{\mathrm{reg}}(\theta,q)>0.
\]
This pair of trees proves that the homogeneous-field phase diagram and
the coexistence supremum are not determined by the lower and upper
branching densities.  It also shows that zero-field coexistence need
not persist at any nonzero homogeneous field.
\end{proof}

\begin{remark}
The suppressed levels have density zero, but every ray meets their
unbounded nonbranching blocks. Thus the comparison preserves growth
invariants while changing the gap geometry and amenability. It is
more specific than the previously known dependence on geometry
illustrated by finite-leaf decorations \cite{SchonmannTanaka1998}.
\end{remark}

Figure~\ref{fig:equal-density} displays the branching patterns and
the different coexistence sets in Theorem~\ref{thm:equal-density}.

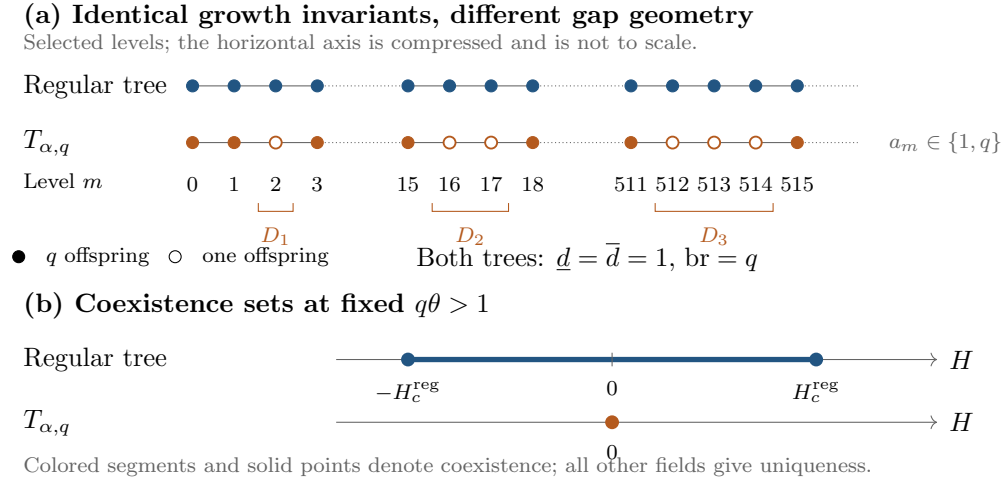
\begin{figure}[htbp]
\centering
\begin{tikzpicture}[x=1cm,y=1cm,font=\small,
  branch/.style={circle,inner sep=0pt,minimum size=4.5pt},
  every node/.style={outer sep=1pt}]
\node[anchor=west,font=\small\bfseries] at (0,3.65)
  {(a) Identical growth invariants, different gap geometry};
\node[anchor=west,font=\scriptsize,text=figgray] at (0,3.29)
  {Selected levels; the horizontal axis is compressed and is not to scale.};
\node[anchor=west] at (0,2.73) {Regular tree};
\node[anchor=west] at (0,1.98) {$T_{\alpha,q}$};
\foreach \yy in {2.73,1.98} {
  \draw[figgray] (2.40,\yy)--(4.05,\yy);
  \draw[figgray,densely dotted] (4.05,\yy)--(5.25,\yy);
  \draw[figgray] (5.25,\yy)--(6.90,\yy);
  \draw[figgray,densely dotted] (6.90,\yy)--(8.20,\yy);
  \draw[figgray] (8.20,\yy)--(10.40,\yy);
  \draw[figgray,densely dotted] (10.40,\yy)--(11.20,\yy);
}
\foreach \xx/\lev in {2.40/0,2.95/1,3.50/2,4.05/3,
  5.25/15,5.80/16,6.35/17,6.90/18,
  8.20/511,8.75/512,9.30/513,9.85/514,10.40/515} {
  \node[branch,fill=figblue,draw=figblue] at (\xx,2.73) {};
  \node[anchor=north,font=\scriptsize] at (\xx,1.69) {$\lev$};
}
\foreach \xx in {2.40,2.95,4.05,5.25,6.90,8.20,10.40}
  \node[branch,fill=figorange,draw=figorange] at (\xx,1.98) {};
\foreach \xx in {3.50,5.80,6.35,8.75,9.30,9.85}
  \node[branch,fill=white,draw=figorange,line width=.7pt] at (\xx,1.98) {};
\node[anchor=west,font=\scriptsize,text=figgray] at (11.45,1.98)
  {$a_m\in\{1,q\}$};
\node[anchor=west,font=\scriptsize] at (0,1.48) {Level $m$};
\draw[figorange] (3.27,1.18)--(3.27,1.07)--(3.73,1.07)--(3.73,1.18);
\draw[figorange] (5.57,1.18)--(5.57,1.07)--(6.58,1.07)--(6.58,1.18);
\draw[figorange] (8.52,1.18)--(8.52,1.07)--(10.08,1.07)--(10.08,1.18);
\node[anchor=north,font=\scriptsize,text=figorange] at (3.50,1.02) {$D_1$};
\node[anchor=north,font=\scriptsize,text=figorange] at (6.075,1.02) {$D_2$};
\node[anchor=north,font=\scriptsize,text=figorange] at (9.30,1.02) {$D_3$};
\node[branch,fill=black,draw=black] at (.10,.47) {};
\node[anchor=west,font=\scriptsize] at (.30,.47) {$q$ offspring};
\node[branch,fill=white,draw=black] at (2.17,.47) {};
\node[anchor=west,font=\scriptsize] at (2.37,.47) {one offspring};
\node[anchor=west,font=\small] at (5.20,.47)
  {Both trees: $\underline d=\overline d=1$, $\operatorname{br}=q$};
\node[anchor=west,font=\small\bfseries] at (0,-.17)
  {(b) Coexistence sets at fixed $q\theta>1$};
\node[anchor=west] at (0,-.89) {Regular tree};
\node[anchor=west] at (0,-1.72) {$T_{\alpha,q}$};
\foreach \yy in {-.89,-1.72} {
  \draw[->,figgray] (4.30,\yy)--(12.25,\yy) node[right,text=black] {$H$};
}
\draw[figblue,line width=2pt] (5.25,-.89)--(10.65,-.89);
\fill[figblue] (5.25,-.89) circle (2.6pt);
\fill[figblue] (10.65,-.89) circle (2.6pt);
\draw[figgray] (7.95,-.98)--(7.95,-.80);
\node[anchor=north,font=\scriptsize] at (5.25,-1.02) {$-H_c^{\mathrm{reg}}$};
\node[anchor=north,font=\scriptsize] at (7.95,-1.02) {$0$};
\node[anchor=north,font=\scriptsize] at (10.65,-1.02) {$H_c^{\mathrm{reg}}$};
\fill[figorange] (7.95,-1.72) circle (2.6pt);
\node[anchor=north,font=\scriptsize] at (7.95,-1.86) {$0$};
\node[anchor=west,font=\scriptsize,text=figgray] at (0,-2.31)
  {Colored segments and solid points denote coexistence; all other fields give uniqueness.};
\end{tikzpicture}
\caption{The equal-density comparison of Theorem~\ref{thm:equal-density}.
Deleting the branching levels
$D_k=\{2^{k^2},\ldots,2^{k^2}+k-1\}$ creates nonbranching gaps of
unbounded length while preserving both branching densities and the
branching number. Panel (a) displays selected levels from the exact
construction. For $q\theta>1$, panel (b) shows the closed regular-tree
coexistence interval and the singleton coexistence set of the modified
tree. In particular, both regular-tree endpoints are included.}
\label{fig:equal-density}
\end{figure}
\FloatBarrier

\section{Quantitative boundary influence and its possible slow decay}
\label{sec:boundary-influence}

For fixed \(\theta\) and a homogeneous field \(H\), write
\(u_m^{\pm,N}(H)\) when the field dependence must be displayed, and set
\begin{equation}\label{eq:boundary-influence}
\Delta_N(H):=
\mu_{N,H/\beta}^{+}(\sigma_o)-\mu_{N,H/\beta}^{-}(\sigma_o)
=\tanh u_0^{+,N}(H)-\tanh u_0^{-,N}(H),\qquad N\ge1.
\end{equation}
Here \(\mu(\sigma_o)\) denotes the expectation of the root spin.
Attractiveness implies \(\Delta_N(H)\ge0\), and spin flip gives
\(\Delta_N(-H)=\Delta_N(H)\).

\begin{theorem}
\label{thm:quantitative-gap}
Let \(H\ne0\), and put
\[
h=|H|,\qquad M_h=h+qK,\qquad c_h=\frac{qK}{h+qK}<1,
\]
\begin{equation}\label{eq:sufficient-corridor-length}
L_h=1+\max\left\{0,
\left\lceil\frac{\log(4K/h)}{|\log\theta|}\right\rceil\right\}.
\end{equation}
Whenever the following set is nonempty, define
\begin{equation}\label{eq:last-sufficient-corridor}
R_h(N)=\max\bigl\{b\in\mathbb Z_{\ge0}:
b+L_h+1\le N,\quad a_{b+1}=\cdots=a_{b+L_h}=1\bigr\}.
\end{equation}
Then
\begin{equation}\label{eq:quantitative-gap-bound}
0\le\Delta_N(H)
\le u_0^{+,N}(H)-u_0^{-,N}(H)
\le M_h\log\frac{M_h}{h}\,c_h^{R_h(N)}.
\end{equation}
Moreover, if \(\ell_j\ge L_h\) and \(N\ge n_{j+1}\),
the right-hand side can be replaced by
\(M_h\log(M_h/h)c_h^{n_j}\).
\end{theorem}

\begin{proof}
It suffices to consider \(H=h>0\), because
\(u_m^{+,N}(-h)=-u_m^{-,N}(h)\).
The fixed point \(r_h\) of \(g_h\) satisfies
\(h\le r_h<h+K\), while \(g_h(-\infty)=h-K\).
Using the Lipschitz constant \(\theta\) of \(g_h\) after this first
iterate, we obtain, for every integer \(L\ge1\),
\[
g_h^L(-\infty)\ge r_h-2K\theta^{L-1}
\ge h-2K\theta^{L-1}.
\]
The choice \eqref{eq:sufficient-corridor-length} therefore ensures
\(g_h^{L_h}(-\infty)\ge h/2>0\).

Suppose that \(b\) belongs to the set in
\eqref{eq:last-sufficient-corridor}. The recursion through its
\(L_h\) one-child levels gives
\[
u_{b+1}^{-,N}(h)
=g_h^{L_h}\bigl(u_{b+L_h+1}^{-,N}(h)\bigr)
\ge g_h^{L_h}(-\infty)>0.
\]
This remains valid when \(b+L_h+1=N\), since the terminal field is
then \(-\infty\). One further recursion step gives
\(u_b^{-,N}(h)>h\). The backward recursion preserves the lower
bound \(h\), and every field after at least one step is at most
\(M_h\). Consequently,
\[
h\le u_m^{-,N}(h)\le u_m^{+,N}(h)\le M_h,
\qquad 0\le m\le b.
\]
The logarithmic contraction from the proof of
Theorem~\ref{thm:unbounded-gaps} now yields
\[
0\le\log u_0^{+,N}(h)-\log u_0^{-,N}(h)
\le c_h^b\log\frac{M_h}{h}.
\]
The exponential function is \(M_h\)-Lipschitz on
\([\log h,\log M_h]\), and \(\tanh\) is \(1\)-Lipschitz on
\(\mathbb R\). These two bounds prove
\eqref{eq:quantitative-gap-bound} with \(b=R_h(N)\).
If \(\ell_j\ge L_h\) and \(N\ge n_{j+1}\), then \(b=n_j\)
belongs to the same set, proving the last assertion.
\end{proof}

The extremal bound also controls arbitrary boundary conditions. If
\(\tau\) and \(\xi\) are two assignments of spins at depth \(N\),
attractiveness gives
\[
\bigl|\mu_{N,H/\beta}^{\tau}(\sigma_o)
-\mu_{N,H/\beta}^{\xi}(\sigma_o)\bigr|\le\Delta_N(H).
\]
The total-variation distance between their root marginals is at most
\(\Delta_N(H)/2\).

For fixed \(\theta\), the sufficient length satisfies
\(L_h=O(\log(1/h))\) as \(h\downarrow0\).
Unbounded gaps imply \(R_h(N)\to\infty\), so the bound also proves
nonzero-field uniqueness. More quantitatively, if
\(R_h(N)\ge\eta N-C\) for all sufficiently large \(N\), with
\(\eta>0\), then
\[
\Delta_N(H)\le
M_h\log(M_h/h)c_h^{-C}
\exp\{-\eta|\log c_h|N\}.
\]
Such spacing information is additional to unboundedness of the gaps.
Indeed, the next result rules out any universal rate of boundary
forgetting for the latter class, even at branching density one.

\begin{theorem}
\label{thm:arbitrarily-slow}
Fix \(q\ge2\), \(\theta\in(1/q,1)\), and a field satisfying
\[
0<|H|\le H_c^{\mathrm{reg}}(\theta,q),
\]
where \(H_c^{\mathrm{reg}}>0\) is the regular-tree coexistence
threshold given in Proposition~\ref{prop:regular-field}. For every sequence
\(\varepsilon_N>0\) with \(\varepsilon_N\to0\), there is a tree
\(T_{\alpha,q}\) whose lower and upper branching densities equal
one, whose branching number is \(q\), and whose nonbranching gaps
are unbounded, such that
\begin{equation}\label{eq:arbitrarily-slow-influence}
\limsup_{N\to\infty}\frac{\Delta_N(H)}{\varepsilon_N}=\infty.
\end{equation}
This tree has a unique Gibbs measure at every nonzero homogeneous
field and every finite inverse temperature. In particular,
\(\Delta_N(H)\to0\) for the parameters in the statement.
\end{theorem}

\begin{proof}
By spin flip we may take \(H>0\). Proposition~\ref{prop:regular-field}
gives distinct least and greatest fixed points \(v^-<v^+\) of
\(Q_H\), including when \(H=H_c^{\mathrm{reg}}\).
Suppose that the child numbers at levels \(0,\ldots,m-1\) have
already been chosen. Their recursion composition is
\[
\mathcal P_m=f_{a_0,H}\circ\cdots\circ f_{a_{m-1},H},
\qquad f_{a,H}(u)=H+aF_\theta(u),
\]
with \(\mathcal P_0\) the identity. Every factor is strictly increasing, so
\[
\delta_m:=\tanh \mathcal P_m(v^+)-\tanh \mathcal P_m(v^-)>0.
\]
If we choose regular branching, \(a_s=q\), for
\(m\le s<N\), then
\[
u_0^{\pm,N}(H)=\mathcal P_m\bigl(Q_H^{N-m}(\pm\infty)\bigr).
\]
Monotonicity and the fixed-point equations imply
\[
Q_H^{N-m}(-\infty)\le v^-<v^+
\le Q_H^{N-m}(+\infty),
\]
and hence \(\Delta_N(H)\ge\delta_m\).

Start with the empty prefix \(m=0\). At stage \(k\ge1\), choose
an integer \(N_k>m\) so large that
\[
N_k\ge k^3,\qquad \varepsilon_{N_k}<\delta_m/k.
\]
Such a choice is possible because \(\varepsilon_N\to0\).
Set \(a_s=q\) for \(m\le s<N_k\), then set
\[
a_s=1\quad(N_k\le s<N_k+k),\qquad a_{N_k+k}=q.
\]
Continue with the new prefix length \(m=N_k+k+1\).
This procedure defines all child numbers and leaves exactly \(k\)
consecutive one-child levels in its \(k\)-th inserted block,
separated from the other such blocks by branching levels. The
nonbranching gaps are therefore unbounded. The finite-volume
recursion with boundary at \(N_k\) depends only on child numbers
at levels below \(N_k\), so later choices do not change it. Thus
\[
\frac{\Delta_{N_k}(H)}{\varepsilon_{N_k}}
\ge\frac{\delta_m}{\varepsilon_{N_k}}>k,
\]
where \(m\) is the prefix length at the start of stage \(k\).
This proves \eqref{eq:arbitrarily-slow-influence}.

For completeness, let \(r_N=N-b_N\) count the one-child levels
below \(N\). Whenever \(N_k\le N<N_{k+1}\),
\[
0\le r_N\le\sum_{i=1}^{k}i=\frac{k(k+1)}2,
\qquad
\frac{r_N}{N}\le\frac{k(k+1)}{2N_k}
\le\frac{k+1}{2k^2}\longrightarrow0.
\]
This estimate applies inside the inserted blocks as well as between
them. Consequently \(b_N/N\to1\), and the branching-number formula
gives \(\br(T_{\alpha,q})=q\).
Finally, Theorem~\ref{thm:unbounded-gaps} gives uniqueness at every
nonzero field for every finite inverse temperature; the convergence
of the extremal fields then implies \(\Delta_N(H)\to0\).
\end{proof}

The two mechanisms in this section are illustrated in
Figure~\ref{fig:boundary-influence}: sufficiently long corridors give
an upper bound on boundary influence, while long regular stretches
can postpone its decay to an arbitrarily large depth.

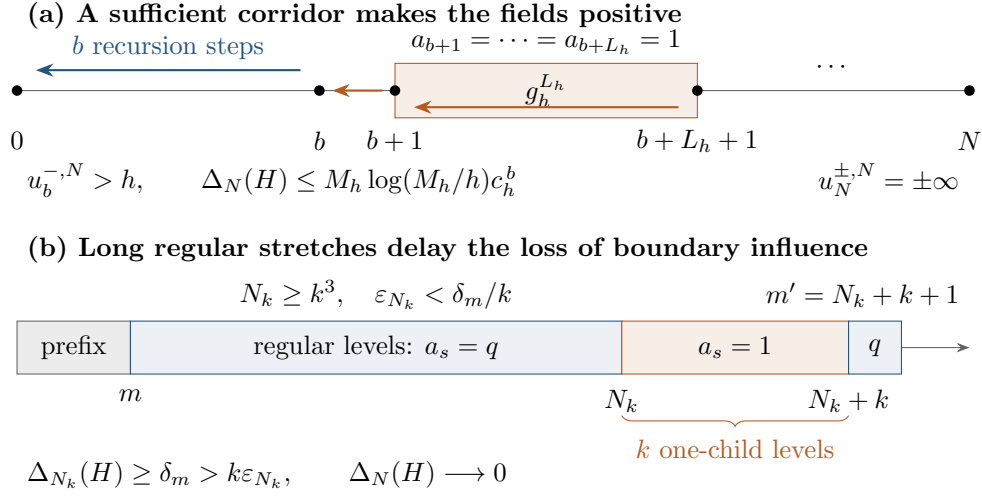
\begin{figure}[htbp]
\centering
\begin{tikzpicture}[x=1cm,y=1cm,>=Stealth,font=\small,
  point/.style={circle,fill=black,inner sep=1.4pt},
  block/.style={draw,rounded corners=1pt,minimum height=7mm}]
\node[anchor=west,font=\small\bfseries] at (0,0)
  {(a) A sufficient corridor makes the fields positive};
\draw[figgray] (0,-1.0)--(12.6,-1.0);
\draw[fill=figorange!10,draw=figorange] (5,-1.36) rectangle (9,-.64);
\node at (7,-1.0) {$g_h^{L_h}$};
\node[align=center] at (7,-.36)
  {$a_{b+1}=\cdots=a_{b+L_h}=1$};
\foreach \x in {0,4,5,9,12.6} \node[point] at (\x,-1.0) {};
\node at (0,-1.66) {$0$};
\node at (4,-1.66) {$b$};
\node at (5,-1.66) {$b+1$};
\node at (9,-1.66) {$b+L_h+1$};
\node at (12.6,-1.66) {$N$};
\draw[->,thick,figblue] (3.75,-.73)--(.25,-.73)
  node[midway,above] {$b$ recursion steps};
\draw[->,thick,figorange] (4.83,-1.0)--(4.17,-1.0);
\draw[->,thick,figorange] (8.75,-1.22)--(5.25,-1.22);
\node at (10.8,-.69) {$\cdots$};
\node[anchor=east] at (12.6,-2.18)
  {$u_N^{\pm,N}=\pm\infty$};
\node[anchor=west] at (0,-2.18)
  {$u_b^{-,N}>h,\qquad \Delta_N(H)\le M_h\log(M_h/h)c_h^{\,b}$};
\node[anchor=west,font=\small\bfseries] at (0,-3.10)
  {(b) Long regular stretches delay the loss of boundary influence};
\draw[fill=figgray!12,draw=figgray] (0,-4.75) rectangle (1.5,-4.05);
\node at (.75,-4.40) {prefix};
\draw[fill=figblue!8,draw=figblue] (1.5,-4.75) rectangle (8,-4.05);
\node at (4.75,-4.40) {regular levels: $a_s=q$};
\draw[fill=figorange!10,draw=figorange] (8,-4.75) rectangle (11,-4.05);
\node at (9.5,-4.40) {$a_s=1$};
\draw[fill=figblue!8,draw=figblue] (11,-4.75) rectangle (11.7,-4.05);
\node at (11.35,-4.40) {$q$};
\draw[->,figgray] (11.7,-4.40)--(12.6,-4.40);
\node[anchor=north] at (1.5,-4.82) {$m$};
\node[anchor=north] at (8,-4.82) {$N_k$};
\node[anchor=north] at (11,-4.82) {$N_k+k$};
\node[anchor=east] at (12.6,-3.70) {$m'=N_k+k+1$};
\node at (4.75,-3.70) {$N_k\ge k^3,\quad \varepsilon_{N_k}<\delta_m/k$};
\draw[decorate,decoration={brace,mirror,amplitude=4pt},figorange]
  (8,-5.30)--(11,-5.30)
  node[midway,below=5pt] {$k$ one-child levels};
\node[anchor=west] at (0,-6.08)
  {$\Delta_{N_k}(H)\ge\delta_m>k\varepsilon_{N_k},\qquad
    \Delta_N(H)\longrightarrow0$};
\end{tikzpicture}
\caption{Two roles of nonbranching stretches (schematic level axes).
In (a), $h=|H|>0$ and any admissible $b$ in
\eqref{eq:last-sufficient-corridor} gives the displayed bound;
choosing the last such $b=R_h(N)$ gives
Theorem~\ref{thm:quantitative-gap}. Arrows indicate the direction of
the backward recursion. The positivity statement is written after
spin flip to $H=h$.
In (b), $q\theta>1$ and $0<|H|\le H_c^{\mathrm{reg}}(\theta,q)$;
rectangles specify offspring numbers on half-open intervals of integer
levels. The boundary at $N_k$ precedes the newly inserted
one-child block, so its root influence still exceeds $k\varepsilon_{N_k}$.
Repeating these stages yields a density-one tree with a unique Gibbs
measure and arbitrarily slow convergence, as in
Theorem~\ref{thm:arbitrarily-slow}.}
\label{fig:boundary-influence}
\end{figure}
\FloatBarrier

\section{Eventually periodic gaps and the full-period return map}
\label{sec:periodic}

We first identify the extremal Gibbs fields with fixed points of a full-period map and then classify those fixed points. Assume that there are integers \(j_0\ge1\) and \(r\ge1\) such that
\[
\ell_{j+r}=\ell_j
\qquad\text{for every }j\ge j_0.
\]
For \(i=0,\ldots,r-1\), put
\[
\ell_i^\circ:=\ell_{j_0+i},
\qquad
d_i:=\ell_i^\circ+1,
\qquad
T_{i,H}:=\Phi_{\ell_i^\circ,H}.
\]
Thus \(d_i\) is the distance between the corresponding consecutive branching levels. Define the one-period return map
\[
\mathcal R_H
:=
T_{0,H}\circ T_{1,H}\circ\cdots\circ T_{r-1,H},
\]
where the rightmost map is applied first, and let
\[
P:=\sum_{i=0}^{r-1}d_i
=
n_{j_0+r}-n_{j_0}.
\]
The branching pattern on the tail is periodic in the level coordinate with period \(P\).

In the constant-gap case \(r=1\), if \(n_{j+1}-n_j=p\) eventually, then \(d_0=p\) and
\[
\mathcal R_H
=
\Phi_{p-1,H}
=
Q_H\circ g_H^{p-1}.
\]

\begin{theorem}
\label{thm:periodic-return}
Under the preceding assumption, every fixed point of \(\mathcal R_H\) determines a tail-periodic spherically symmetric splitting Gibbs measure. The limiting branch fields of the minus and plus states at level \(n_{j_0}\) are respectively the least and greatest fixed points of \(\mathcal R_H\). Moreover, 
\[
\mu_B^+=\mu_B^-
\quad\Longleftrightarrow\quad
\mathcal R_H\text{ has a unique fixed point}.
\]
\end{theorem}

\begin{proof}
Each map \(T_{i,H}\), and hence \(\mathcal R_H\), is continuous and strictly increasing on \(\mathbb R\), extends continuously to \([-\infty,+\infty]\), and maps the extended line into the compact interval \(I_H\). For \(M\ge1\), set
\[
x_M^-:=\mathcal R_H^{\,M}(-\infty),
\qquad
x_M^+:=\mathcal R_H^{\,M}(+\infty).
\]
Monotonicity gives
\[
x_M^-\le x_{M+1}^-,
\qquad
x_{M+1}^+\le x_M^+.
\]
Both sequences are bounded in \(I_H\), so they converge to limits \(x^-\le x^+\). Continuity gives
\[
\mathcal R_H(x^-)=x^-,
\qquad
\mathcal R_H(x^+)=x^+.
\]
If \(z\) is any fixed point of \(\mathcal R_H\), then
\[
x_M^-\le z\le x_M^+
\qquad\text{for every }M,
\]
and hence \(x^-\le z\le x^+\). Thus \(x^-\) and \(x^+\) are the least and greatest fixed points.

Choose the finite-volume depths
\[
N_M:=n_{j_0+Mr}.
\]
Starting from the terminal fields \(\pm\infty\) at level \(N_M\) and iterating the block recursion through \(M\) complete periods shows that the fields at level \(n_{j_0}\) are exactly \(x_M^\pm\). Since the full plus and minus finite-volume limits exist, this subsequence identifies their branch fields at \(n_{j_0}\) as \(x^-\) and \(x^+\).

Conversely, let \(z\) be a fixed point of \(\mathcal R_H\). Set \(z_r:=z\) and define successively
\[
z_i:=T_{i,H}(z_{i+1}),
\qquad i=r-1,r-2,\ldots,0.
\]
The fixed-point equation gives \(z_0=z_r=z\). Therefore
\[
v_{j_0+mr+i}:=z_i,
\qquad
m\ge0,\qquad 0\le i\le r-1,
\]
satisfies every block recursion on the periodic tail. Filling each gap by the scalar recursion and propagating through the finite initial segment produces a full boundary-law profile. Since both the gap word and the phase vector repeat after \(P\) levels, this profile is tail-periodic. The boundary-law construction gives the corresponding splitting Gibbs measure.

If \(\mathcal R_H\) has more than one fixed point, then \(x^-<x^+\). The finite composition from level \(n_{j_0}\) to the root is strictly increasing, so the minus and plus root fields are different and hence \(\mu_B^-\ne\mu_B^+\).

If \(\mathcal R_H\) has a unique fixed point, then \(x^-=x^+\), so the two limiting profiles agree at \(n_{j_0}\). To propagate this equality along the tail, write
\[
v_j^\pm=\Phi_{\ell_j,H}(v_{j+1}^\pm).
\]
Every block map is strictly increasing and therefore injective. Thus \(v_{j_0}^-=v_{j_0}^+\) implies \(v_{j_0+1}^-=v_{j_0+1}^+\), and induction gives equality at every later branching level. The intermediate fields and the finite initial segment are then uniquely determined by the recursion. Therefore the complete plus and minus field profiles coincide, and so do the corresponding Gibbs measures.
\end{proof}

\subsection{The complete periodic phase diagram}

The full-period reduction admits a global fixed-point classification.
The key observation is that the Schwarzian derivative remains negative
under every composition occurring in the return map.

\begin{lemma}
\label{lem:periodic-schwarzian}
For every \(H\in\mathbb R\), the map \(\mathcal R_H\) has at most
three distinct fixed points. If it has three, all three are simple
zeros of \(\mathcal R_H(u)-u\).
\end{lemma}

\begin{proof}
For a \(C^3\) map \(f\) with \(f'>0\), write
\[
Sf=\frac{f'''}{f'}-\frac32\left(\frac{f''}{f'}\right)^2.
\]
The chain rule
\(S(f\circ g)=(Sf\circ g)(g')^2+Sg\), applied to
\(F_\theta=\arctanh\circ(\theta\tanh)\), gives
\begin{equation}\label{eq:schwarzian-cavity}
SF_\theta(u)=2\bigl(F_\theta'(u)^2-1\bigr)<0,
\end{equation}
because \(0<F_\theta'\le\theta<1\). Adding a constant or
multiplying the output by a nonzero constant leaves the Schwarzian
unchanged. Thus every elementary map \(u\mapsto H+aF_\theta(u)\)
has negative Schwarzian, and the composition rule yields
\(S\mathcal R_H<0\). Consequently
\[
h_H:=(\mathcal R_H')^{-1/2},\qquad
h_H''=-\tfrac12(S\mathcal R_H)h_H>0.
\]
Strict convexity implies that \(h_H=1\), equivalently
\(\mathcal R_H'=1\), has at most two solutions. Rolle's theorem
now bounds the number of zeros of \(\mathcal R_H(u)-u\) by three.
If three distinct zeros included a multiple zero, the latter and the
two open intervals between consecutive zeros would provide three
distinct zeros of \(\mathcal R_H'-1\), a contradiction.
\end{proof}

\begin{theorem}
\label{thm:periodic-phase}
Put \(A=q^r\theta^P\). If \(A\le1\), the Gibbs measure is unique
at every homogeneous field \(H\), and \(\mathcal R_H\) has exactly
one fixed point.

If \(A>1\), there is a unique number \(H_c\in(0,qK)\) such that
\begin{equation}\label{eq:periodic-fixed-point-count}
\#\operatorname{Fix}(\mathcal R_H)=
\begin{cases}
3,& |H|<H_c,\\
2,& |H|=H_c,\\
1,& |H|>H_c.
\end{cases}
\end{equation}
In particular, the set of homogeneous fields with more than one
Gibbs measure is exactly \([-H_c,H_c]\).

At \(H=H_c\), there is a unique negative fixed point \(u_c\), and
\begin{equation}\label{eq:periodic-fold}
\mathcal R_{H_c}(u_c)=u_c,\qquad
\mathcal R_{H_c}'(u_c)=1,\qquad
\mathcal R_{H_c}''(u_c)>0,\qquad
\left.\partial_H\mathcal R_H(u_c)\right|_{H=H_c}>0.
\end{equation}
The other fixed point is positive and simple. Thus the nonzero
endpoints are nondegenerate saddle-node bifurcations, related by
spin-flip symmetry. More precisely, write the three fixed points for
\(0<H<H_c\) as \(w_1(H)<w_2(H)<w_3(H)\). Then
\begin{equation}\label{eq:periodic-square-root}
\begin{aligned}
w_1(H)&=u_c-\gamma_c\sqrt{H_c-H}+O(H_c-H),\\
w_2(H)&=u_c+\gamma_c\sqrt{H_c-H}+O(H_c-H),\\
\gamma_c&=
\left(
\frac{2\left.\partial_H\mathcal R_H(u_c)\right|_{H=H_c}}
{\mathcal R_{H_c}''(u_c)}
\right)^{1/2}>0.
\end{aligned}
\end{equation}
Here the merging branches are the least and the middle fixed points;
the greatest fixed point \(w_3(H)>0\) remains simple.
\end{theorem}

\begin{proof}
Write \(D_H(u)=\mathcal R_H(u)-u\). At \(H=0\), each block map
is \(T_{i,0}(u)=qF_\theta^{\circ d_i}(u)\). Hence
\[
\mathcal R_0'(0)=A,\qquad
0<\mathcal R_0(u)<Au\quad(u>0),
\]
and \(\mathcal R_0\) is odd. If \(A\le1\), zero is therefore
the only fixed point at zero field. If \(A>1\), then
\(D_0(u)>0\) for sufficiently small positive \(u\), whereas
\(D_0(u)\to-\infty\) as \(u\to+\infty\). Oddness and
Lemma~\ref{lem:periodic-schwarzian} give exactly three simple
fixed points: one negative, zero, and one positive.

Fix \(H>0\). Every elementary map sends \([0,\infty)\) into
\((0,\infty)\), is increasing, and is strictly concave on
\((0,\infty)\). Their composition \(\mathcal R_H\) has the same
properties and satisfies \(\mathcal R_H(0)>0\). Since its range
is bounded, it has a positive fixed point. At any such point \(z\),
concavity gives
\[
\mathcal R_H'(z)
\le\frac{\mathcal R_H(z)-\mathcal R_H(0)}{z}<1.
\]
The strict concavity of \(D_H\), together with \(D_H(0)>0\),
also shows that this positive fixed point is unique. It is simple.

Differentiating the elementary maps with respect to \(H\) gives
one, while their derivatives with respect to their arguments are
positive. The chain rule therefore yields
\begin{equation}\label{eq:periodic-field-monotonicity}
\partial_H\mathcal R_H(u)\ge1
\quad\text{for all }u,H.
\end{equation}
When \(A\le1\), oddness and the zero-field inequality imply
\(D_0(u)>0\) for \(u<0\). Thus \(D_H(u)>0\) for \(u\le0\)
and \(H>0\). The positive fixed point is the only one. The identity
\(\mathcal R_{-H}(-u)=-\mathcal R_H(u)\) handles negative fields,
proving the first assertion by Theorem~\ref{thm:periodic-return}.

Suppose henceforth that \(A>1\), and define
\[
m(H)=\min_{u\le0}D_H(u),\qquad H\ge0.
\]
Since \(\mathcal R_H(u)\in[H-qK,H+qK]\), we have
\(D_H(u)\ge H-qK-u\). This bound tends to \(+\infty\) as
\(u\to-\infty\), uniformly for \(H\) in compact intervals.
Comparing with the bounded values \(D_H(0)\) shows that all
minimizers lie in one compact interval whenever \(H\) ranges over
a fixed compact interval. Thus the minimum is attained and \(m\)
is continuous. Equation~\eqref{eq:periodic-field-monotonicity}
implies
\[
m(H_2)\ge m(H_1)+(H_2-H_1)\qquad(H_2>H_1).
\]
Since \(D_0'(0)=A-1>0\), we have \(m(0)<0\). At \(H=qK\),
the outermost elementary map is \(qK+qF_\theta\), which is
strictly positive at every finite argument. Hence
\(D_{qK}(u)>0\) for all \(u\le0\), and attainment gives
\(m(qK)>0\). There is therefore a unique \(H_c\in(0,qK)\)
with \(m(H_c)=0\).

For \(0<H<H_c\), the inequalities \(D_H(0)>0\),
\(D_H(u)\to+\infty\) as \(u\to-\infty\), and \(m(H)<0\)
give two distinct negative zeros. Along with the positive zero,
these are exactly three simple fixed points by
Lemma~\ref{lem:periodic-schwarzian}. At \(H=H_c\), a minimizing
point \(u_c<0\) is a negative zero with \(D_{H_c}'(u_c)=0\).
Together with the simple positive zero it gives exactly two fixed
points: a third would contradict the final assertion of that lemma.
For \(H>H_c\), \(m(H)>0\), so only the positive fixed point
remains. Zero-field symmetry and spin-flip symmetry establish
\eqref{eq:periodic-fixed-point-count};
Theorem~\ref{thm:periodic-return} gives the coexistence statement.

It remains to prove nondegeneracy. At the negative minimum \(u_c\),
\(\mathcal R_{H_c}'(u_c)=1\) and
\(\mathcal R_{H_c}''(u_c)\ge0\). If the latter were zero, the
strictly convex function \(h_{H_c}=(\mathcal R_{H_c}')^{-1/2}\)
would satisfy \(h_{H_c}(u_c)=1\) and \(h_{H_c}'(u_c)=0\).
It would follow that \(h_{H_c}(u)>1\), and therefore
\(D_{H_c}'(u)<0\), for every \(u\ne u_c\). This would make
\(D_{H_c}\) strictly decreasing, contradicting its two distinct
zeros. Hence \(\mathcal R_{H_c}''(u_c)>0\); positivity of the
field derivative follows from
\eqref{eq:periodic-field-monotonicity}.

All the maps are real analytic. In the Taylor expansion of
\(D_H(u)\) at \((u_c,H_c)\), the leading terms are
\[
\left.\partial_H\mathcal R_H(u_c)\right|_{H=H_c}(H-H_c)
+\tfrac12\mathcal R_{H_c}''(u_c)(u-u_c)^2.
\]
Substituting \(H=H_c-t^2\) and \(u=u_c+tz\), dividing by
\(t^2\), and applying the implicit-function theorem at
\(t=0\), \(z=\pm\gamma_c\), proves
\eqref{eq:periodic-square-root}. The positive fixed point remains
simple by the concavity argument above.
\end{proof}

Figure~\ref{fig:periodic-bifurcation} illustrates the classification
using the regular-tree return map. In particular, the fold at a
nonzero critical field leaves two distinct extremal fixed points,
so coexistence persists at the endpoint.

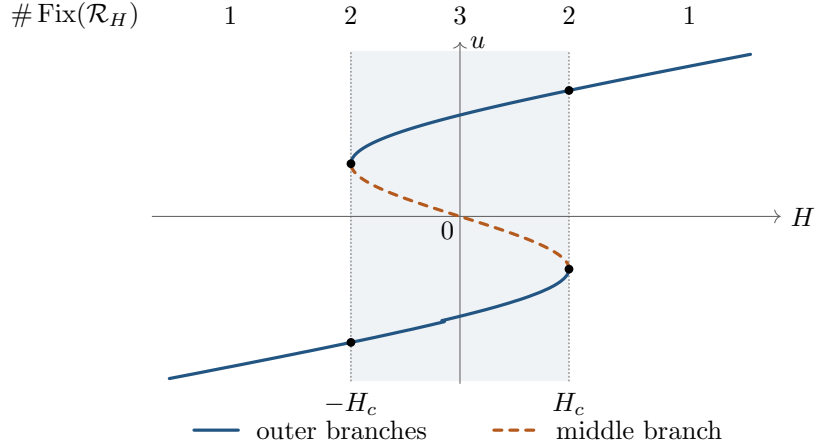
\begin{figure}[htbp]
\centering
\begin{tikzpicture}[x=3.45cm,y=0.65cm,font=\small,
  line cap=round,line join=round]
  \def\figHc{0.4180480043}
  \def\figUc{1.0729483047}
  \def\figOther{2.5639446137}
  \fill[figblue!7] (-\figHc,-3.35) rectangle (\figHc,3.35);
  \draw[figgray,densely dotted] (-\figHc,-3.35)--(-\figHc,3.35);
  \draw[figgray,densely dotted] (\figHc,-3.35)--(\figHc,3.35);
  \draw[->,figgray] (-1.18,0)--(1.23,0) node[right,text=black] {$H$};
  \draw[->,figgray] (0,-3.4)--(0,3.55) node[right,text=black] {$u$};
  \node[below left,inner sep=2pt] at (0,0) {$0$};
  \draw[figblue,very thick]
    plot[domain=-3.3:-\figUc,samples=101,variable=\t]
    ({\t-ln((1.8*exp(2*\t)+0.2)/(0.2*exp(2*\t)+1.8))},{\t});
  \draw[figorange,very thick,dashed]
    plot[domain=-\figUc:\figUc,samples=101,variable=\t]
    ({\t-ln((1.8*exp(2*\t)+0.2)/(0.2*exp(2*\t)+1.8))},{\t});
  \draw[figblue,very thick]
    plot[domain=\figUc:3.3,samples=101,variable=\t]
    ({\t-ln((1.8*exp(2*\t)+0.2)/(0.2*exp(2*\t)+1.8))},{\t});
  \fill (-\figHc,\figUc) circle[radius=1.65pt];
  \fill (-\figHc,-\figOther) circle[radius=1.65pt];
  \fill (\figHc,-\figUc) circle[radius=1.65pt];
  \fill (\figHc,\figOther) circle[radius=1.65pt];
  \node[below] at (-\figHc,-3.38) {$-H_c$};
  \node[below] at (\figHc,-3.38) {$H_c$};
  \node[anchor=east] at (-1.19,4.1) {$\#\operatorname{Fix}(\mathcal R_H)$};
  \node at (-0.88,4.1) {$1$};
  \node at (-\figHc,4.1) {$2$};
  \node at (0,4.1) {$3$};
  \node at (\figHc,4.1) {$2$};
  \node at (0.88,4.1) {$1$};
  \draw[figblue,very thick] (-1.02,-4.35)--(-0.85,-4.35);
  \node[anchor=west] at (-0.82,-4.35) {outer branches};
  \draw[figorange,very thick,dashed] (0.13,-4.35)--(0.30,-4.35);
  \node[anchor=west] at (0.33,-4.35) {middle branch};
\end{tikzpicture}
\caption{The fixed-point diagram in Theorem~\ref{thm:periodic-phase},
illustrated by the regular tree with \(q=2\) and \(\theta=0.8\).
The plotted curve is the exact relation
\(H=u-2\arctanh(0.8\tanh u)\), with
\(H_c\simeq0.418048\). The shaded interval includes both critical
fields. At \(H=H_c\), the lower and middle branches meet at a fold,
while the upper fixed point remains distinct; the reflected statement
holds at \(-H_c\). The four filled dots display the two fixed points
at each endpoint. The counts refer to fixed points of the return map,
not to the number of all Gibbs measures.}
\label{fig:periodic-bifurcation}
\end{figure}
\FloatBarrier

\begin{corollary}
\label{cor:periodic-critical}
The zero-field threshold is \(q^r\theta^P=1\). At equality the
unique fixed point is zero, and it has exact multiplicity three:
\begin{equation}\label{eq:periodic-critical-cubic}
\mathcal R_0(u)-u=-C_{\mathrm{per}}u^3+O(u^5),
\qquad C_{\mathrm{per}}>0.
\end{equation}
More explicitly, set
\[
A_i=q\theta^{d_i},\qquad
c_i=\frac{A_i(1-\theta^{2d_i})}{3}.
\]
Then
\[
C_{\mathrm{per}}
=\sum_{i=0}^{r-1}c_i
\left(\prod_{k=0}^{i-1}A_k\right)
\left(\prod_{k=i+1}^{r-1}A_k^3\right).
\]
Empty products equal one. 
\end{corollary}

\begin{proof}
The identity
\(F_\theta^{\circ d}(u)=\arctanh(\theta^d\tanh u)\) gives
\[
T_{i,0}(u)=A_i u-c_i u^3+O(u^5).
\]
Composing these expansions in the order defining \(\mathcal R_0\)
gives linear coefficient \(\prod_iA_i=A\) and cubic coefficient
\(-C_{\mathrm{per}}\). Every \(c_i\) and \(A_i\) is positive,
so \(C_{\mathrm{per}}>0\). At \(A=1\) this proves
\eqref{eq:periodic-critical-cubic}, and uniqueness follows from
Theorem~\ref{thm:periodic-phase}.
\end{proof}

The periodic tail has branching density \(r/P\), so its threshold
also reads \(\theta q^{r/P}=1\). Since \(b_N-(r/P)N=O(1)\),
the critical series in Theorem~\ref{thm:exact} diverges, in agreement
with the unique critical fixed point.

\subsection{The regular-tree benchmark}
The classical regular-tree formula makes the comparison in
Theorem~\ref{thm:equal-density} explicit; see also
\cite{Georgii2011,SchonmannTanaka1998}.

\begin{proposition}
\label{prop:regular-field}
On the rooted \(q\)-ary tree the Gibbs measure is unique at every
field if \(q\theta\le1\). If \(q\theta>1\), put
\[
s_*=\sqrt{\frac{q\theta-1}{\theta(q-\theta)}},\qquad
u_*=\arctanh s_*,\qquad
H_c^{\mathrm{reg}}(\theta,q)=q\arctanh(\theta s_*)-\arctanh s_*>0.
\]
The number of distinct fixed points of \(u=H+qF_\theta(u)\) is
one for \(|H|>H_c^{\mathrm{reg}}\), two for
\(|H|=H_c^{\mathrm{reg}}\), and three for
\(|H|<H_c^{\mathrm{reg}}\). 
\end{proposition}

\begin{proof}
Write the equation as \(H=G(u)=u-qF_\theta(u)\). If
\(q\theta\le1\), then \(G'\ge0\), with equality possible only
at \(u=0\) when \(q\theta=1\). Hence \(G\) is strictly increasing
from \(-\infty\) to \(+\infty\). If \(q\theta>1\), solving
\(qF_\theta'(u)=1\) gives exactly \(u=\pm u_*\), since
\(0<s_*^2<1\). The derivative of \(G\) is positive outside
\([-u_*,u_*]\) and negative inside. Oddness gives its local maximum
\(G(-u_*)=H_c^{\mathrm{reg}}>0\) and local minimum
\(G(u_*)=-H_c^{\mathrm{reg}}\). This proves the fixed-point counts.
At an endpoint, the tangential fixed point
\(-\operatorname{sgn}(H)u_*\) is double because
\(F_\theta''(\pm u_*)\ne0\); the other fixed point is simple.
Now, we apply Theorem~\ref{thm:periodic-return} with a one-step period to
obtain the assertion for all Gibbs measures.
\end{proof}

\section{Uniform coexistence for bounded gaps}
\label{sec:bounded-gaps}

The unbounded-gap theorem has a quantitative converse at sufficiently
low temperature. The sharp comparison within a class with a prescribed
maximum gap is with the periodic tree having that gap at every step.
For \(L\in\mathbb Z_{\ge0}\), write
\[
\alpha^{(L)}=\{k(L+1):k\ge0\},\qquad
\mathcal C_\theta^{(L)}=\mathcal C_\theta(\alpha^{(L)}),\qquad
H_c^{(L)}=H_c(\theta;\alpha^{(L)},q),
\]
where \(\mathcal C_\theta(\alpha)\) displays the dependence of the
coexistence set on the branching sequence. The return map of
\(T_{\alpha^{(L)},q}\) is \(\Phi_{L,H}\).

\begin{theorem}
\label{thm:bounded-gap-comparison}
Suppose that \(\ell_j\le L\) for all sufficiently large \(j\). Then
\begin{equation}\label{eq:bounded-gap-inclusion}
\mathcal C_\theta^{(L)}\subseteq\mathcal C_\theta(\alpha).
\end{equation}
If \(q\theta^{L+1}\le1\), then
\(\mathcal C_\theta^{(L)}=\varnothing\) and \(H_c^{(L)}=0\).
If \(q\theta^{L+1}>1\), then \(H_c^{(L)}>0\) and
\[
[-H_c^{(L)},H_c^{(L)}]
=\mathcal C_\theta^{(L)}\subseteq\mathcal C_\theta(\alpha).
\]
In particular both displayed endpoints give coexistence. The set
\(\mathcal C_\theta^{(L)}\) is the largest set of fields for which
coexistence can be guaranteed uniformly over all eventually
\(L\)-bounded gap sequences.
\end{theorem}

\begin{proof}
The assertions about \(\mathcal C_\theta^{(L)}\), including its
endpoints, follow from Theorem~\ref{thm:periodic-phase}, applied to
the one-block period. It remains to prove the inclusion.

Let \(H>0\) belong to \(\mathcal C_\theta^{(L)}\).
The map \(\Phi_{L,H}\) has a unique positive fixed point.
Indeed, it maps \([H,H+qK]\) into itself, and each elementary
recursion map \(u\mapsto H+aF_\theta(u)\), \(a\in\{1,q\}\),
contracts logarithmic distances on this interval by the factor
\(qK/(H+qK)<1\), as in the proof of
Theorem~\ref{thm:unbounded-gaps}. A positive fixed point must lie in
this interval. Also \(\Phi_{L,H}(0)>0\), so zero is not fixed.
Theorem~\ref{thm:periodic-return} therefore supplies a negative
fixed point \(-t\), with \(t>0\).

Since \(g_H(-t)>-t\) and \(g_H\) is increasing, its iterates from
\(-t\) increase. Consequently, for every \(0\le\ell\le L\),
\[
\Phi_{\ell,H}(-t)\le\Phi_{L,H}(-t)=-t.
\]
Spin-flip symmetry and monotonicity in the field give
\[
\Phi_{\ell,H}(t)\ge\Phi_{\ell,-H}(t)
=-\Phi_{\ell,H}(-t)\ge t.
\]
Set \(M=qK+H\). Every block map has range in
\([H-qK,H+qK]\subseteq[-M,M]\), and \(t\le M\) because
\(-t\) is a fixed point. Thus \([-M,-t]\) and \([t,M]\)
are disjoint compact invariant intervals for every \(\ell\le L\).
Proposition~\ref{prop:invariant-intervals} yields coexistence on the
eventually bounded sequence; its proof also propagates distinct
profiles through any finite initial segment. Negative fields follow
by spin flip.

At \(H=0\), coexistence on the constant-gap tree entails
\(q\theta^{L+1}>1\), and there is a positive fixed point \(t\)
of \(\Phi_{L,0}\). The inequality
\(F_\theta(u)<u\) for \(u>0\) gives
\(\Phi_{\ell,0}(t)\ge\Phi_{L,0}(t)=t\) for \(\ell\le L\).
Oddness gives the negative inequality, and the same interval
argument applies. Finally, the constant-gap sequence itself belongs
to the class being compared. No field outside its coexistence set
can therefore be guaranteed uniformly over this class.
\end{proof}

The comparison is exact but uses the coexistence endpoint of a
periodic map. The following lower bound is explicit.

\begin{corollary}
\label{cor:bounded-gap-explicit}
Suppose \(\ell_j\le L\) eventually and \(q\theta^{L+1}>1\). Put
\begin{equation}\label{eq:bounded-gap-constants}
\lambda=\theta^{L+1},\qquad
D_L=1+q\theta\frac{1-\theta^L}{1-\theta},\qquad
s_*=\sqrt{\frac{q\lambda-1}{\lambda(q-\lambda)}}.
\end{equation}
Then
\begin{equation}\label{eq:bounded-gap-explicit}
0<h_*:=\frac{q\arctanh(\lambda s_*)-\arctanh s_*}{D_L}
=\frac{H_c^{\mathrm{reg}}(\lambda,q)}{D_L}
\le H_c^{(L)},
\end{equation}
and every \(|H|\le h_*\) gives coexistence on \(T_{\alpha,q}\).
For \(L=0\), this bound is the exact regular-tree endpoint.
\end{corollary}

\begin{proof}
The chain rule and \(0<F_\theta'\le\theta\) give, uniformly in
\(u,H\),
\[
0<\partial_H\Phi_{\ell,H}(u)
\le1+q\theta\sum_{k=0}^{\ell-1}\theta^k\le D_L
\qquad(0\le\ell\le L),
\]
where the sum is empty when \(\ell=0\). Hence
\(|\Phi_{\ell,H}(u)-\Phi_{\ell,0}(u)|\le D_L|H|\).
At zero field,
\[
\Phi_{\ell,0}(u)=qF_\theta^{\circ(\ell+1)}(u),\qquad
F_\theta^{\circ(L+1)}(u)=\arctanh(\lambda\tanh u).
\]
By Proposition~\ref{prop:regular-field} with parameter \(\lambda\),
the function \(qF_\theta^{\circ(L+1)}(t)-t\), \(t\ge0\),
attains its positive maximum
\(H_c^{\mathrm{reg}}(\lambda,q)\) at
\(t_*=\arctanh s_*>0\). Since \(F_\theta(t)<t\) for \(t>0\),
for every \(\ell\le L\) and \(|H|\le h_*\) we have
\[
\Phi_{\ell,H}(t_*)
\ge qF_\theta^{\circ(L+1)}(t_*)-D_L|H|\ge t_*.
\]
Oddness at zero field similarly gives
\(\Phi_{\ell,H}(-t_*)\le-t_*\). With \(M=qK+|H|\),
the intervals \([-M,-t_*]\) and \([t_*,M]\) are invariant;
here \(t_*<qK\) follows from
\(qF_\theta^{\circ(L+1)}(t_*)-t_*>0\).
Proposition~\ref{prop:invariant-intervals} proves coexistence,
including at equality. Applying this also to the constant-gap tree
gives \(h_*\le H_c^{(L)}\). For \(L=0\), \(\lambda=\theta\)
and \(D_L=1\).
\end{proof}

\begin{corollary}
\label{cor:gap-dichotomy}
For \(T_{\alpha,q}\), the following statements are equivalent:
\begin{enumerate}[label=\textup{(\roman*)}]
\item \(\sup_j\ell_j<\infty\);
\item the underlying graph is nonamenable;
\item for some \(0<\beta<\infty\) and \(B\ne0\), there is more
than one Gibbs measure.
\end{enumerate}
Here \(J>0\) is fixed. More quantitatively, an eventual gap bound
\(L\) guarantees a nontrivial interval of coexistence at every
temperature satisfying \(q\tanh(\beta J)^{L+1}>1\).
\end{corollary}

\begin{proof}
If the gaps are unbounded, arbitrarily long paths consisting of
one-child vertices have exactly two boundary edges. They form a
F\o lner sequence, so the graph is amenable. If the gaps are bounded,
let
\(M=\max\{n_1,\sup_j\ell_j\}<\infty\). For a finite connected
vertex set \(S\), let \(b(S)\) count its branching vertices and
let \(e(S)\) count its boundary edges. The tree degree identity
gives \(e(S)\ge(q-1)b(S)+1\). Every component of the
nonbranching vertices of \(S\) is a path of at most \(M\) vertices;
following it towards descendants assigns it either to a branching
vertex in \(S\) or to a boundary edge, with no repetitions. Thus
\[
|S|\le b(S)+M\bigl(b(S)+e(S)\bigr)
\le\frac{1+qM}{q-1}\,e(S).
\]
The same bound holds componentwise for arbitrary finite sets.
The edge-isoperimetric constant is therefore positive, proving
(i)\(\Leftrightarrow\)(ii).

If (i) holds, choose \(L\) bounding the gaps and then a finite
\(\beta\) large enough that \(q\tanh(\beta J)^{L+1}>1\).
Corollary~\ref{cor:bounded-gap-explicit} gives a positive \(H\)
with coexistence, and \(B=H/\beta\ne0\). Conversely,
Theorem~\ref{thm:unbounded-gaps} excludes (iii) when (i) fails.
\end{proof}

The existence of a nonzero-field phase transition on a bounded-degree
nonamenable graph is already part of the general theorem of
Jonasson and Steif~\cite{JonassonSteif1999}. The additional content
here is the comparison at a fixed temperature and maximum gap,
together with the explicit field interval
\eqref{eq:bounded-gap-explicit}. The converse in this particular
non-quasi-transitive family follows from the unbounded-gap theorem.

\section{Supplementary criteria for general gap sequences}
\label{sec:sufficient}

We finish with two supplementary uniqueness bounds and the elementary
invariant-interval criterion used in the bounded-gap comparison.

\subsection{A field-dependent derivative bound}
For \(I_H=[H-qK,H+qK]\), let
\[
d_H=\dist(0,I_H)=\max\{0,|H|-qK\},\qquad
\rho_H=\sup_{u\in I_H}F_\theta'(u)=F_\theta'(d_H).
\]
Then \(0<\rho_H\le\theta\) and \(\rho_H\to0\) as
\(|H|\to\infty\).

\begin{proposition}
\label{thm:high-field-uniqueness}
If \(\rho_Hq^{\underline d}<1\), the Gibbs measure is unique.
If \(\rho_Hq^{\overline d}<1\), there is an \(r\in(0,1)\) such that
for every fixed \(k\) there is a finite \(C_k\) with
\[
0\le u_k^{+,N}-u_k^{-,N}\le C_kr^N
\]
for all sufficiently large \(N\). Both conclusions hold for all
sufficiently large \(|H|\).
\end{proposition}

\begin{proof}
Put \(\Delta_m^{(N)}=u_m^{+,N}-u_m^{-,N}\). The first step from
the boundary gives \(\Delta_{N-1}^{(N)}\le2qK\). At every earlier
step both input fields lie in \(I_H\), so the mean value theorem gives
\(\Delta_m^{(N)}\le a_m\rho_H\Delta_{m+1}^{(N)}\). Consequently,
\begin{equation}\label{eq:field-influence-bound}
\Delta_k^{(N)}\le2qK\rho_H^{N-k-1}q^{b_{N-1}-b_k}.
\end{equation}
If \(\rho_Hq^{\underline d}<1\), choose \(\varepsilon>0\) with
\(r_-:=\rho_Hq^{\underline d+\varepsilon}<1\). Infinitely many
\(N\) satisfy \(b_{N-1}\le(\underline d+\varepsilon)(N-1)\).
Along those depths,
\(\Delta_k^{(N)}\le2qK\rho_H^{-k}q^{-b_k}r_-^{N-1}\to0\).
The depth monotonicity of the extremal fields identifies their full
limits, so Proposition~\ref{prop:field-recursion} gives uniqueness.

Under the upper-density assumption, choose \(\varepsilon>0\) with
\(r:=\rho_Hq^{\overline d+\varepsilon}<1\). Now
\(b_{N-1}\le(\overline d+\varepsilon)(N-1)\) for all large \(N\),
and \eqref{eq:field-influence-bound} gives the assertion with
\(C_k=2qK\rho_H^{-k}q^{-b_k}r^{-1}\). Finally,
\(\overline d\le1\) and \(\rho_H\to0\) give the large-field claim.
\end{proof}

\subsection{Products of block contractions}
\begin{proposition}
\label{prop:gap-contraction}
Set \(L_{\ell,H}=\Lip(\Phi_{\ell,H}|_{I_H})\). If
\(\prod_{j=1}^M L_{\ell_j,H}\to0\) as \(M\to\infty\), the Gibbs
measure is unique. For \(H\ne0\), let \(r_H\) be the unique
fixed point of \(g_H\), and set
\(\lambda_H=F_\theta'(r_H)<\theta\). For each
\(\lambda\in(\lambda_H,1)\), there is a finite constant
\(C_{\lambda,\theta,H}\), independent of \(q\), such that
\[
L_{\ell,H}\le C_{\lambda,\theta,H}q\theta\lambda^\ell,
\qquad \ell\ge0.
\]
\end{proposition}

\begin{proof}
The limiting branch fields belong to \(I_H\), so for every \(M\),
\[
0\le v_1^+-v_1^-
\le\left(\prod_{j=1}^M L_{\ell_j,H}\right)
(v_{M+1}^+-v_{M+1}^-)
\le2qK\prod_{j=1}^M L_{\ell_j,H}.
\]
Equality at the first branching level implies uniqueness by
Proposition~\ref{prop:field-recursion}.

Since \(0<g_H'\le\theta<1\), the map \(g_H\) has a unique fixed
point \(r_H\). It is nonzero when \(H\ne0\), and
\eqref{eq:derivative} gives \(\lambda_H<\theta\). Choose a closed
neighborhood \(U\) of \(r_H\) with \(g_H(U)\subset U\) and
\(\sup_Ug_H'\le\lambda\). The first iterate of \(I_H\) lies in
\(J_H=[H-K,H+K]\), independently of \(q\). Uniform convergence of
\(g_H^n\) to \(r_H\) on \(J_H\) gives an integer \(n_0\), depending
only on \(\lambda,\theta,H\), such that
\(g_H^{n_0}(I_H)\subset U\). The chain rule therefore gives
\(\Lip(g_H^\ell|_{I_H})\le\theta^{n_0}\lambda^{\ell-n_0}\) for
\(\ell\ge n_0\). Enlarging a constant to include the finitely many
smaller \(\ell\), and using \(\Lip(Q_H)\le q\theta\), proves the bound.
\end{proof}

The product condition can use the frequency of moderate gaps.
Theorem~\ref{thm:unbounded-gaps} requires no such frequency control
when the gaps are unbounded.

\subsection{Common invariant intervals}
\begin{proposition}
\label{prop:invariant-intervals}
Suppose there are \(j_0\) and two disjoint compact intervals
\(I^-=[a^-,b^-]\), \(I^+=[a^+,b^+]\), with \(b^-<a^+\), such that
\[
\Phi_{\ell_j,H}(I^-)\subset I^-,\qquad
\Phi_{\ell_j,H}(I^+)\subset I^+
\quad\text{for all }j\ge j_0.
\]
Then there are at least two splitting Gibbs measures, and hence
\(\mu_B^-\ne\mu_B^+\).
\end{proposition}

\begin{proof}
For \(M\ge j_0\), choose terminal values
\(w_{M+1}^{\pm,M}\in I^\pm\) and set
\(w_j^{\pm,M}=\Phi_{\ell_j,H}(w_{j+1}^{\pm,M})\) for
\(j=M,\ldots,j_0\). Invariance keeps these fields in their respective
compact intervals. A diagonal subsequence as \(M\to\infty\) gives
\(w_j^\pm\in I^\pm\) with
\(w_j^\pm=\Phi_{\ell_j,H}(w_{j+1}^\pm)\) for all \(j\ge j_0\).
For \(n_j<m\le n_{j+1}\), set
\(\widetilde u_m^\pm=g_H^{n_{j+1}-m}(w_{j+1}^\pm)\), and at
branching levels set \(\widetilde u_{n_j}^\pm=w_j^\pm\).
These definitions agree at shared endpoints and satisfy
\eqref{eq:cavity}. Propagating through the finite initial segment
gives two full profiles and hence two Gibbs measures. Their fields at
\(n_{j_0}\) are strictly ordered. The finite composition carrying
them to the root is strictly increasing, so their root marginals
differ. The minimal and maximal Gibbs measures must therefore differ.
\end{proof}

For a finite set of allowed gaps only finitely many interval
inclusions need checking. For a periodic sequence,
Theorem~\ref{thm:periodic-return} gives the sharper exact criterion.

\section{Conclusion}

Nonzero-field coexistence is controlled by the gap sequence, which cannot be detected by branching density or number. Unbounded gaps provide uniqueness in nonzero homogenous fields. Bounded gaps allow for coexistence at limited temperatures and nonzero fields. The periodic tree with the biggest allowable gap provides the ideal coexistence interval, which is guarantyd uniformly throughout all ultimately bounded gap sequences.

The full-period map provides both an accurate Gibbs uniqueness criteria and a comprehensive fixed-point categorization for gaps that inevitably occur.
The nonzero-field coexistence endpoints are nondegenerate folds due to the negative Schwarzian derivative, which implies only three fixed points. The square-root law illustrates how two branches of boundary law merge.

Finite-volume estimations indicate how lengthy nonbranching sections reduce border influence. For given nonzero-field parameters in the regular tree's coexistence area, the density-one construction eliminates uniform rates across all density-one trees with unbounded gaps.
The precise coexistence interval for a limited aperiodic gap sequence, including endpoint behavior, is still unknown beyond the uniform comparison and adequate conditions shown here.

\appendix

\section{Direct zero-field proof and critical profiles}
\label{sec:zero-field-appendix}

We give a short direct proof of the known spherically symmetric criterion
in Theorem~\ref{thm:exact}, specialized from
\cite[Corollary~2.4]{PemantlePeres2010}, and record its critical consequences.
Throughout this appendix \(H=0\), \(0<\theta<1\), and
\(P_N=\theta^Nq^{b_N}\), with \(P_0=1\).

\begin{proof}[Direct proof of Theorem~\ref{thm:exact}]
For \(t>0\), put
\[
 \Psi_\theta(t)=\frac1{F_\theta(t)^2}-\frac1{\theta^2t^2}.
\]
Since
\(F_\theta'(t)=\theta(1-\tanh^2t)/(1-\theta^2\tanh^2t)<\theta\)
for \(t>0\), we have \(0<F_\theta(t)<\theta t\), so
\(\Psi_\theta(t)>0\).  The expansion
\(F_\theta(t)=\theta t-\theta(1-\theta^2)t^3/3+O(t^5)\)
and the limit \(F_\theta(+\infty)=\arctanh\theta\) give
\[
 \Psi_\theta(0)=\frac{2(1-\theta^2)}{3\theta^2},
 \qquad
 \Psi_\theta(+\infty)=\frac1{(\arctanh\theta)^2}.
\]
Thus \(\Psi_\theta\) extends to a continuous positive function on
\([0,+\infty]\), and constants \(0<c_\theta\le C_\theta<\infty\)
bound it from below and above.

Let \(h_m^{(N)}=u_m^{+,N}\) denote the plus cavity fields at zero
field, so \(h_N^{(N)}=+\infty\) and
\(h_m^{(N)}=a_mF_\theta(h_{m+1}^{(N)})\).
Set \(r_m^{(N)}=(h_m^{(N)})^{-2}\), with \((+\infty)^{-2}=0\).
The preceding identity, including its value at infinity, yields
\begin{align*}
 r_m^{(N)}
  &=(a_m\theta)^{-2}r_{m+1}^{(N)}
       +a_m^{-2}\Psi_\theta(h_{m+1}^{(N)}),\\
 r_0^{(N)}
  &=\theta^2\sum_{\ell=1}^N
           P_\ell^{-2}\Psi_\theta(h_\ell^{(N)}).
\end{align*}
Consequently,
\begin{equation}\label{eq:reciprocal-square-bounds}
 \theta^2c_\theta\sum_{\ell=1}^NP_\ell^{-2}
 \le\frac1{(h_0^{(N)})^2}
 \le\theta^2C_\theta\sum_{\ell=1}^NP_\ell^{-2}.
\end{equation}
If the series converges, these bounds keep \(h_0^{(N)}\) away from
zero.  The limiting plus root magnetization is then positive, and
spin-flip symmetry gives distinct plus and minus Gibbs measures.

If the series diverges, apply the same calculation from any fixed
level \(k\).  Indeed, with
\(P_{k,n}=\prod_{m=k}^{k+n-1}(a_m\theta)=P_{k+n}/P_k\),
\[
 \sum_{n\ge1}P_{k,n}^{-2}
   =P_k^2\sum_{\ell>k}P_\ell^{-2}=\infty,
 \qquad
 \frac1{(h_k^{(N)})^2}
   \ge\theta^2c_\theta\sum_{n=1}^{N-k}P_{k,n}^{-2}.
\]
Hence \(u_k^+=0=u_k^-\) for every \(k\).  The reconstruction following
Proposition~\ref{prop:field-recursion} identifies the plus and minus
Gibbs measures, and their extremal ordering gives uniqueness among
all DLR Gibbs measures.
\end{proof}

\begin{corollary}
\label{cor:zero-field-profiles}
Write \(d=\underline d(\alpha)\).  At zero field there is coexistence
if \(\theta q^d>1\) and uniqueness if \(\theta q^d<1\).
If \(d=0\), uniqueness holds at every finite inverse temperature.
If \(d>0\), the critical inverse temperature is
\[
 \beta_c=J^{-1}\arctanh(q^{-d}),
 \qquad
 \text{coexistence at }\beta_c
 \quad\Longleftrightarrow\quad
 \sum_{N\ge1}q^{-2(b_N-dN)}<\infty.
\]
In particular, the critical Gibbs measure is unique when \(d=1\).
For \(0<d<1\), a profile
\(b_N=dN+c\log N+O(1)\) gives critical coexistence exactly when
\(c>1/(2\log q)\).
\end{corollary}

\begin{proof}
If \(\theta q^d>1\), choose \(\varepsilon>0\) such that
\(\theta q^{d-\varepsilon}>1\).  Eventually
\(b_N\ge(d-\varepsilon)N\), so the series in
Theorem~\ref{thm:exact} converges geometrically.  If
\(\theta q^d<1\), choose \(\varepsilon>0\) such that
\(\theta q^{d+\varepsilon}<1\).  Infinitely many \(N\) satisfy
\(b_N\le(d+\varepsilon)N\), along which \(P_N^{-2}\) does not tend to
zero.  This proves the strict regimes and, since \(\theta<1\) at finite
temperature, the assertion for \(d=0\).  For \(d>0\), substitute
\(\theta=q^{-d}\) in the exact series to obtain the critical criterion.
When \(d=1\), the bound \(b_N\le N\) makes every critical summand at
least one.  Finally, the logarithmic profile makes the critical series
comparable to \(\sum_{N\ge1}N^{-2c\log q}\), including divergence at
\(2c\log q=1\).
\end{proof}

\begin{example}
\label{ex:critical-log-profile}
Every \(d\in(0,1)\) and \(c\in\mathbb R\) in the preceding corollary
can be realized.  Choose an integer \(M\ge1\) so large that
\(0<d+c\log(1+1/(N+M))<1\) for all \(N\ge0\), and define
\[
 b_N=\lfloor dN+c\log(N+M)\rfloor-\lfloor c\log M\rfloor,
 \qquad
 \alpha=\{N\ge0:b_{N+1}-b_N=1\}.
\]
Then \(b_0=0\), all increments belong to \(\{0,1\}\), and
\(b_N=\#(\alpha\cap\{0,\ldots,N-1\})
      =dN+c\log N+O(1)\) as \(N\to\infty\).
Because \(d>0\), the branching set is infinite.
This is the familiar polynomial correction to critical exponential
level growth in the Pemantle--Peres criterion, expressed in terms of
admissible sparse branching levels.
\end{example}

For comparison, \(n_j\sim cj\), \(c\ge1\), gives \(d=1/c\) and
\(\beta_c=J^{-1}\arctanh(q^{-1/c})\).  At \(c=1\) critical uniqueness
is forced; at \(c>1\) sublinear counting fluctuations can affect
critical behavior.  The profiles \(n_j\sim j^\gamma\), \(\gamma>1\),
and \(n_j\sim e^{cj}\), \(c>0\), both give \(d=0\) and therefore
zero-field uniqueness at every finite temperature.

In our next studies, we are going to investigate the set $\textrm{ex}(\mathcal{G})\setminus\textrm{ex}(\mathcal{G}_{TI})$, and the absolute continuity of these measures. In the case of regular trees such investigates have been preformed in \cite{AMSa14,I,MBS162,MK23,MS20}.

\section*{Statements and Declarations}

\subsection*{Funding}
This work is not supported by any fund.

\subsection*{Competing Interests}
The authors have no relevant financial or non-financial interests
to disclose.

\subsection*{Author Contributions}
All authors contributed equally
to the conception, mathematical analysis, and preparation of the
manuscript. Both authors reviewed and approved the final manuscript.

\subsection*{Data Availability}
No datasets were generated or analysed during the current study.

\end{document}